\documentclass[a4paper,11pt]{article}
\usepackage{authblk}
\usepackage{fullpage}
\usepackage[letterpaper, margin=1in]{geometry}
\usepackage{hyperref}

\title{Efficient Stable Population Protocols for Parity and Beyond}
\author[1]{Leszek Gąsieniec}
\author[2]{Tytus Grodzicki} 
\author[2]{Tomasz Jurdzi{\'n}ski\thanks{Supported by the National Science Centre, Poland under project number 2020/39/B/ST6/03288.}}
\author[2]{\\Jakub Kowalski\thanks{Supported in part by the National Science Centre, Poland under project number  2021/41/B/ST6/03691.}} 
\author[2]{Grzegorz Stachowiak$^*$}
\affil[1]{School of Computer Science and Informatics, University of Liverpool}
\affil[2]{Institute of Computer Science, University of Wroc{\l}aw }
\date{}

\usepackage{lineno}
\usepackage{amsfonts,amsmath,amssymb,amsthm}
\usepackage{tcolorbox}
\usepackage{multicol}
\usepackage{environ}
\usepackage[symbol]{footmisc}

\usepackage[T1]{fontenc}
\usepackage{graphicx}
\usepackage[dvipsnames]{xcolor}
\usepackage[utf8]{inputenc}
\usepackage{calc}
\usepackage{thmtools}
\usepackage{thm-restate}

\newtheorem{theorem}{Theorem}
\newtheorem{lemma}{Lemma}
\newtheorem{claim}{Claim}

\newcommand{\mathescape}[1]{\relax\ifmmode#1\else $#1$\fi} % todo add forced space when in no-math mode?

\usepackage{algorithm,algorithmicx}
\usepackage[noend]{algpseudocode}
\MakeRobust{\Call}
 
\newcommand{\LineIf}[2]{\State \algorithmicif\ {#1}\ \algorithmicthen\ {#2}}% \algorithmicend\ \algorithmicif}
\newcommand{\indLineIf}[2]{\State \hspace{\algorithmicindent} \algorithmicif\ {#1}\ \algorithmicthen\ {#2}}% \algorithmicend\ \algorithmicif}
\newcommand{\DLineIf}[2]{\State \algorithmicif\ {#1}\ \algorithmicthen\ \State \hspace{\algorithmicindent} {#2}}% \algorithmicend\ \algorithmicif}
\newcommand{\LineIfElse}[3]{\State \algorithmicif\ {#1}\ \algorithmicthen\ {#2} \algorithmicelse\ {#3}}% \algorithmicend\ \algorithmicif}

\newcommand{\indLineIfElse}[3]{\State \hspace{\algorithmicindent} \algorithmicif\ {#1}\ \algorithmicthen\ {#2} \State \hspace{\algorithmicindent} \algorithmicelse\ {#3}}
\newcommand{\indindLineIfElse}[3]{\State \hspace{\algorithmicindent}\hspace{\algorithmicindent} \algorithmicif\ {#1}\ \algorithmicthen\ {#2} \State \hspace{\algorithmicindent}\hspace{\algorithmicindent} \algorithmicelse\ {#3}}

\newcommand{\indindLineIfElseElse}[5]{\State \hspace{\algorithmicindent}\hspace{\algorithmicindent} \algorithmicif\ {#1}\ \algorithmicthen\ {#2} \State \hspace{\algorithmicindent}\hspace{\algorithmicindent} \algorithmicelse\ \algorithmicif\ {#3}\ \algorithmicthen\ {#4} \State \hspace{\algorithmicindent}\hspace{\algorithmicindent} \algorithmicelse\ {#5}
}
\newcommand{\indLineIfEIEIEI}[8]{\State \hspace{\algorithmicindent} \algorithmicif\ {#1}\ \algorithmicthen\ {#2} \State \hspace{\algorithmicindent} \algorithmicelse\ \algorithmicif\ {#3}\ \algorithmicthen\ {#4} \State \hspace{\algorithmicindent} \algorithmicelse\ \algorithmicif\ {#5}\ \algorithmicthen\ {#6} \State \hspace{\algorithmicindent} \algorithmicelse\ \algorithmicif\ {#7}\ \algorithmicthen\ {#8}
}

\algdef{SE}[UNT]{Unt}{EndUnt}[1]{\textbf{until}\ #1\ \algorithmicdo}{\algorithmicend\ \textbf{until}}%

\definecolor{light-gray}{gray}{0.95} % matches tcolorbox background
\usepackage{listings}
\lstdefinestyle{mystyle}{
    basicstyle=\ttfamily\small, 
    commentstyle=\color{gray},
    backgroundcolor=\color{light-gray},
    morecomment=[f][\color{grey}][0]{\#},
    morecomment=[f][\color{grey}][0]{//},
    numbers=left,                    
    numbersep=5pt,    
    mathescape=true,
    keywordstyle=\textbf,
    morekeywords={while,if,then,else,true,false,do,:,parallel,in,def,until},
}
\newtheorem{fact}{Fact}

\begin{document}

\maketitle

\thispagestyle{empty}

\begin{abstract}

The {\em population protocol} model studies {\em pairwise interactions} among simple, indistinguishable {\em agents}. Each agent has limited storage, represented by a single state from a fixed space. In the {\em probabilistic variant} considered here, a {\em random scheduler} selects an {\em interacting} pair uniformly at random at each step.
Beyond {\em state utilisation}, we consider the {\em stabilisation time}, defined as the number of interactions to reach a final configuration divided by the population size $n$. A protocol is {\em time-efficient} if it stabilises in polylogarithmic time, {\em space-efficient} if it uses polylogarithmic states, {\em stable} if it converges to the correct answer with probability $1$, and {\em silent} if no states change after stabilisation.

%The {\em population protocol} model is used to study the power of {\em pairwise interactions} between simple, indistinguishable entities, referred to as {\em agents}. In this model, each agent is equipped with limited storage, represented by a single state drawn from a predefined state space. In the {\em probabilistic variant} of population protocols adopted here, a {\em random scheduler} selects an {\em interacting} pair of agents uniformly at random at each step of the protocol.
%%In this variant, in addition to efficient {\em state utilisation}, one is also concerned with the {\em stabilisation time},
%determined by the number of interactions required to reach the final configuration, divided by the population size $n$.
%We say that a protocol is {\em time-efficient} if it stabilises in polylogarithmic (in $n$) time, and {\em space-efficient} if it utilises polylogarithmic state space. We also say that a protocol is {\em stable} if it stabilises with the correct answer with probability 1, and {\em silent} if agents do not change their states after stabilisation.

For nearly two decades, {\em population protocols} have been extensively investigated, producing efficient solutions for central problems in distributed computing, notably {\em majority} computation, one of the two primary predicates in {\em Presburger arithmetic}. Yet the complementary class of {\em congruence} predicates has remained substantially more elusive. %~\cite{AAD+06,DBLP:journals/dc/AngluinAER07}.
This long‑standing gap constitutes a major open challenge in the area. In particular, while time‑efficient (non‑stable) solutions exist via Monte Carlo simulation of register machines%~\cite{AngluinAE08}
, all known stable protocols for such predicates remain slow, exhibiting either exponential%~\cite{blondin_et_al:LIPIcs.STACS.2020.40,DBLP:journals/acta/EsparzaGLM17}
\ or highly polynomial time complexity. %~\cite{DBLP:journals/jcss/CzernerGHE24}.
In this paper, we present a systematic methodology that enables the design of protocols that are {\em efficient}, {\em stable} and {\em silent} for computing {\em congruence predicates}.
We first address the {\em parity problem}, where agents compute the parity of a designated sub‑population size, and then generalise the approach to congruence modulo an arbitrary integer $m$.

Our congruence 
%\gst{(congruence?)} 
protocols achieve efficiency through a modular construction, emphasising universality, robustness, and strong probabilistic guarantees.
While prior work focused on problem‑specific designs, our approach provides a systematic framework for efficient multi‑stage protocols.
Central to this approach is the {\em weight system}, an efficient counting mechanism based on the novel concept of {\em population weights} (or simply {\em weights}), combined with {\em majority} computation and operating via implicit conversion between unary and binary representations. The weight system is integrated within the proposed modular MC+ computing paradigm, which leverages extensive prior work on {\em clocking mechanisms}, efficient {\em anomaly signalling}, and a {\em switching mechanism} that guarantees fallback to a slow yet always‑correct stable protocol.
The new approach yields the first efficient parity and congruence protocols using $O(\log^3 n)$ states and achieving silent stabilisation in $O(\log^3 n)$ time, %addressing an open problem posed by 
providing significant progress towards the research directions identified by Czerner {\em et al.} (2024). %in~\cite{DBLP:journals/jcss/CzernerGHE24}.
We also discuss the impact of the implicit conversion  
%enabled by the weight system, 
on other problems, including the computation and alternative representation of (sub-)population sizes, a postulate raised by Doty and Eftekhari (2021).% in~\cite{DOTY202191}.

All algorithmic solutions presented in this paper guarantee silent stabilisation in time both expected and with high probability (whp) defined as $1-n^{-\eta},$ 
for a constant $\eta>0.$
\end{abstract}

%the proof of Theorem~\ref{th:parity} Section~\ref{s:fullparity}, of Theorem~\ref{majority_mc+} in Section~\ref{s:fullmajority}, the definition and the analysis of the slow congruence protocol in Section~\ref{subsec:slow}, and pseudocodes of protocols Algorithm 2 ($\Call{LeaderElection}{}$), Algorithm 3 ($\Call{WeightCreation}{}$), Algorithm 4 ($m$-$\Call{Congruence}{}$), and Algorithm 5 ($m$-$\Call{WeightCreation}{}$), in Section~\ref{s:pseudocodes}.

\newpage

% \TODO{unanonymized authors}
% \TODO{dopisać granty}

\setcounter{page}{1}

\section{Introduction} 
The model of population protocols originates from the seminal paper \cite{AAD+06}, providing tools suitable for the formal analysis of {\em pairwise interactions} between simple, indistinguishable entities referred to as {\em agents}. These agents are equipped with limited storage, communication, and computation capabilities.
When two agents engage in a direct interaction, their states change according to the predefined {\em transition function}, which forms an integral part of the population protocol.
It is assumed that a protocol starts in the predefined {\em initial configuration} of agents’ states representing the input, and it {\em stabilises} in one of the {\em final configurations} of states representing the {\em solution} to the considered problem.
In~the standard {\em probabilistic variant} of population protocols adopted here,
in each step of a protocol, the {\em random scheduler} selects an ordered pair of agents: the {\em initiator} and the {\em responder}, which are drawn from the whole population uniformly at random. The lack of symmetry in this pair is a powerful source of randomness utilised by population protocols.
In the probabilistic variant, in addition to efficient {\em state utilisation}, one is also interested in the {\em time complexity}, where the {\em sequential time} refers to the number of interactions leading to the stabilisation of a protocol in a final configuration. More recently, the focus has shifted to the {\em parallel time}, or simply the {\em time}, defined as the sequential time divided by the size $n$ of the whole population. The (parallel) time reflects on the parallelism of simultaneous independent interactions of agents utilised in {\em efficient population protocols} that stabilise in time 
$O(\mbox{polylog}\,n)$.
 
In this paper, the rules of a transition function are denoted as $A + B \rightarrow A' + B'$, meaning that if there is an interaction initiated by any agent in state $A$, with any agent in state $B$, then the former changes its state to $A'$, while the latter changes its state to $B'$.
We say that an event occurs {\em with high probability (whp)} if its probability is at least $1 - n^{-\eta}$, for any fixed constant $\eta > 0$; equivalently, its complement occurs with {\em negligible} probability, i.e., at most $n^{-\eta}$.
All algorithmic solutions presented in this paper guarantee silent stabilisation (always correct solution in which agents do not change their states after stabilisation) in time both expected and whp.

\subsection{Context, Motivation and Results}

The original population-protocol model~\cite{AAD+06} assumes a fixed set of states and transitions, independent of the population size~$n$. However, for fundamental tasks such as {\em leader election}~\cite{DS18} and {\em majority computation}~\cite{DBLP:conf/soda/AlistarhAEGR17}, two central tools in our computing paradigm, as well as for the evaluation of various functions and predicates including parity (congruence)~\cite{BDS17}, stable fixed-state protocols significantly faster than linear are not feasible in leaderless populations. This limitation has motivated models in which the number of states may grow with~$n$, enabling substantially more efficient solutions.

Within this setting, {\em population protocols} have been extensively studied for nearly two decades, yielding efficient solutions to several central problems in distributed computing, most notably {\em majority} computation, one of the two primary predicates in {\em Presburger arithmetic}. In contrast, the complementary class of {\em congruence} predicates has remained substantially more elusive~\cite{AAD+06,DBLP:journals/dc/AngluinAER07}. While time-efficient non-stable solutions can be obtained via Monte Carlo simulation of register machines~\cite{AngluinAE08}, all known stable protocols for such predicates remain slow, exhibiting either exponential~\cite{blondin_et_al:LIPIcs.STACS.2020.40,DBLP:journals/acta/EsparzaGLM17} or highly polynomial time complexity~\cite{DBLP:journals/jcss/CzernerGHE24}. Closing this gap therefore remains a major open challenge.

At the same time, most efficient stable population protocols have been developed for individual problems. Such protocols typically succeed with high probability but may fail with negligible probability; a standard remedy is to switch to a slower but stable backup protocol, as in the leader-election and majority literature discussed in Section~\ref{s:tools}. We take this recurring technique a step further by elevating it to a defining feature of a general framework, which we term {\em MC+ protocols}. The MC+ paradigm supports the sequential composition of efficient protocol stages, together with anomaly detection and fallback to stable computation, thereby providing a systematic framework for constructing efficient multi-stage stable and silent population protocols; see Section~\ref{s:MC+}.

{\bf Our Results} We propose a modular framework for designing efficient and stable congruence protocols, centred on a novel counting mechanism that combines {\em population weights} with {\em majority computation} within the MC+ computing paradigm. The framework builds on prior work on clocking, anomaly signalling, and switching mechanisms that guarantee fallback to stable and silent protocols (Section~\ref{s:MC+}).
This framework yields the first efficient, stable, and silent protocols for both parity (Theorem~\ref{th:parity}) and congruence (Theorem~\ref{th:congruence}), each using $O(\log^3 n)$ states and stabilising in $O(\log^3 n)$ time, both in expectation and with high probability (Section~\ref{s:results}), providing significant progress towards the research directions identified in~\cite{DBLP:journals/jcss/CzernerGHE24}.
We conclude in Section~\ref{s:conclusion} by discussing briefly the impact of the implicit unary-to-binary conversion enabled by the weight system and its applications to other problems, including the computation and exact representation of (sub-)population sizes, a postulate raised in~\cite{DOTY202191}.

{\bf Further Context} In population protocols, integers are represented in unary, as cardinalities of subpopulations in designated states, making addition efficient and stable. Although fixed-state Monte Carlo protocols can simulate a restricted form of {\em register machines} recognising LOGSPACE problems~\cite{AngluinAE08}, efficient stable subtraction, multiplication, and division remain challenging. Subtraction, in particular, requires pairing agents from two subpopulations, potentially taking linear time when their sizes are similar. Majority protocols overcome an analogous difficulty by amplifying the cardinality difference (bias)~\cite{AngluinAE08}. The same work gives an efficient Monte Carlo subtraction method based on iteratively generating powers of two, together with addition and majority. This motivates our {\em system of weights}, which provides an implicit binary representation of integers and associated operations.

%{\bf Further Context} In the population protocol model, integer values are represented in unary format, where each value corresponds to the cardinality of a subpopulation of agents in a designated state. 
%This representation makes addition straightforward, i.e., efficient and stable. 
%While it is well understood that fixed-state Monte Carlo population protocols can simulate a limited variant of {\em register machines} (recognising problems in LOGSPACE)~\cite{AngluinAE08}, 
%the efficient and stable computation of basic operations such as subtraction (multiplication and division) remains significantly more challenging.
%The primary obstacle in subtraction lies in the need for effective pairing of agents from two subpopulations, which may require linear time when their cardinalities are comparable or equal. 
%A similar challenge arises in majority protocols, where it is addressed by amplifying the difference (bias) between the cardinalities to a magnitude that population protocols can process efficiently~\cite{AngluinAE08}. 
%In the same paper, the authors propose a method for efficient Monte Carlo type subtraction by iteratively generating powers of two, combined with simple addition and more sophisticated majority protocols. 
%This approach serves as a prototype for our {\em system of weights}, which introduces an implicit binary representation of integer values, supported by corresponding operations. 

\subsection{Overview of Tools}\label{s:tools}

In this section, we present an overview of the key concepts underpinning both the proposed computing paradigm and the algorithmic solutions.

{\bf Epidemic Process} A \emph{(one-way) epidemic process}~\cite{AngluinAE08} is a simple yet effective protocol
in which each agent can be in one of two possible states, denoted as $0$ (\emph{susceptible}) or $1$ (\emph{infected}), with the transition function consisting of a single meaningful rule
$(x) + (y) \rightarrow (x) + (\max\{x, y\})$.

\begin{lemma}[Lemma 2 in \cite{AngluinAE08}]
    The time required for all agents to become infected, starting from a configuration with a single infected agent, is $\Theta(\log n)$ whp.
\end{lemma}

The epidemic process establishes a fundamental $\Theta(\log n)$-time communication bound in population protocols.
It lies at the core of numerous algorithmic tools and solutions in population protocols, including phase clocks~\cite{DBLP:conf/soda/AlistarhAG18,AngluinAE08,DBLP:journals/jacm/GasieniecS21,DBLP:conf/focs/DotyEGSUS21,DBLP:conf/stacs/GasieniecSS23} and, in turn, efficient leader election~\cite{DBLP:conf/icalp/AlistarhG15,DBLP:conf/podc/BilkeCER17,DBLP:conf/soda/AlistarhAG18, DBLP:conf/soda/BerenbrinkKKO18,DBLP:journals/jacm/GasieniecS21,DBLP:conf/spaa/GasieniecSU19, DBLP:conf/stoc/BerenbrinkGK20}, majority computation~\cite{alistarh2015fastExactMajority,DBLP:conf/soda/AlistarhAEGR17,DBLP:conf/soda/AlistarhAG18,DBLP:conf/podc/BilkeCER17,MocquardAABS2015,mocquard2016optimal,Berenbrink_2020,ben2020log3,DBLP:conf/focs/DotyEGSUS21}, and others. Recent work~\cite{GKL+24} introduces the $k$-contact epidemic model, requiring $k$ distinct infected contacts for infection. In this paper, the epidemic process is also used by the switching mechanism to activate the slow stable protocol: upon receiving an anomaly signal, each agent adopts its corresponding initial state in the slow protocol.

{\bf MC+ Paradigm} While fallback mechanisms for efficient stable population protocols are well explored in the literature, including for {\em leader election}~\cite{DBLP:conf/icalp/AlistarhG15,DBLP:conf/podc/BilkeCER17,DBLP:conf/soda/AlistarhAG18, DBLP:conf/soda/BerenbrinkKKO18,DBLP:journals/jacm/GasieniecS21,DBLP:conf/spaa/GasieniecSU19, DBLP:conf/stoc/BerenbrinkGK20} and {\em majority computation}~\cite{alistarh2015fastExactMajority,DBLP:conf/soda/AlistarhAEGR17,DBLP:conf/soda/AlistarhAG18,DBLP:conf/podc/BilkeCER17,MocquardAABS2015,mocquard2016optimal,Berenbrink_2020,ben2020log3,DBLP:conf/focs/DotyEGSUS21}, no clean “off‑the‑shelf” universal solution for the stable (and silent) composition of multiple efficient protocols is currently known.
To address this gap, we propose the MC+ computing paradigm; for details, see Section~\ref{s:MC+}. Building on prior work, the paradigm has three key properties: \emph{composability}, robust error detection with guaranteed recovery, and efficient execution.
Error detection is based on identifying unusual protocol behaviour that may indicate an incorrect execution. Such behaviour occurs with negligible probability and is referred to as an \emph{anomaly}. Composability and efficient execution are supported by restricting each protocol's operations to a prescribed number of rounds counted by a clock.
Together, these features provide a practical framework for constructing complex, yet easier-to-design and analyse, efficient stable and, silent protocols.

In this paper, the MC+ paradigm is applied to the parity problem to realise Algorithm~\ref{alg:parity} and is subsequently extended to the congruence problem within the population protocol framework in Algorithm~\ref{alg:congruence}. While the construction relies on several key components of the MC+ paradigm, including leader election, clocking and fallback mechanisms, the central novelty lies in the effective integration of the proposed weight system with majority computation. This combination forms the core counting mechanism, enabling efficient and coordinated computation within the protocol.

{\bf Clocks and Synchronisation} Efficient population protocols typically use {\em phase clocks} to synchronise computation. 
Most such clocks are designed to count $\Theta(\log n)$ parallel time, which, assisted by epidemic propagation, is sufficient to fully synchronise all agents in the population. 
Leader- or junta-based phase clocks operate with a fixed number of states~\cite{AngluinAE08,DBLP:journals/jacm/GasieniecS21}, 
and when nested, they can count any polylogarithmic duration. 
An alternative class consists of leaderless clocks that utilise $\Theta(\log n)$ states~\cite{DBLP:conf/soda/AlistarhAG18,DBLP:conf/focs/DotyEGSUS21}, 
enabling approximate parallel-time counting either as fixed-resolution clocks~\cite{DBLP:conf/focs/DotyEGSUS21} or as oscillators~\cite{DBLP:conf/soda/AlistarhAG18}. 
A more recent clock type, introduced in~\cite{DBLP:conf/stacs/GasieniecSS23}, supports counting $\Theta(n\log n)$ parallel time using a fixed number of states, 
leveraging either leaders or network connections in the network constructor model.

Here, synchronisation is supported by a {\em round-counting clock} maintained by each agent, reinforced by fine-grained interaction counting performed by the leader and by epidemic waves.
This precise time-stamping system enables protocols with efficient {\em anomaly signalling} through execution-time inconsistencies, essential for the MC+ protocols (Section~\ref{s:clock}).

{\bf Efficient Leader Election} A central challenge in distributed computing is symmetry breaking, typically achieved by electing one agent as {\em the leader}.
In the context of population protocols, the results in~\cite{DBLP:conf/wdag/ChenCDS14, DBLP:conf/soda/Doty14} laid the foundation for proving that leader election cannot be solved in sublinear time using a fixed number of states~\cite{DS18}.
Alistarh and Gelashvili~\cite{DBLP:conf/icalp/AlistarhG15} then introduced an upper bound with a protocol stabilising in $O(\log^3 n)$ time using $O(\log^3 n)$ states per agent.
Later, Alistarh {\em et al.}~\cite{DBLP:conf/soda/AlistarhAEGR17} explored general trade-offs between state and time complexity, establishing a separation between {\em slowly stabilising} protocols using $o(\log\log n)$ states and {\em rapidly stabilising} ones requiring $\Omega(\log\log n)$ states.
Subsequent advances improved upper bounds:
Bilke {\em et al.}~\cite{DBLP:conf/podc/BilkeCER17} achieved $O(\log^2 n)$ time whp with $O(\log^2 n)$ states;
Alistarh {\em et al.}~\cite{DBLP:conf/soda/AlistarhAG18} and Berenbrink {\em et al.}~\cite{DBLP:conf/soda/BerenbrinkKKO18} reduced states to $O(\log n)$ while preserving $O(\log^2 n)$ time whp using synthetic coins;
G\k asieniec and Stachowiak~\cite{DBLP:journals/jacm/GasieniecS21} further lowered state usage to $O(\log\log n)$ with the same time bound whp;
G\k asieniec {\em et al.}~\cite{DBLP:conf/spaa/GasieniecSU19} traded high probability for expected $O(\log n \log\log n)$ time;
and finally, Berenbrink {\em et al.}~\cite{DBLP:conf/stoc/BerenbrinkGK20} achieved the optimal expected $O(\log n)$ time.
Related work in graphical population protocols can be found in~\cite{Ryb1}.

In our time-efficient protocols, the elected leader initiates weight creation and coordinates rounds and synchronisation. Our leader-election mechanism builds on coin-flip reduction across successive rounds~\cite{AngluinAE08, DBLP:journals/jacm/GasieniecS21, DBLP:conf/stoc/BerenbrinkGK20}, with round numbering ensuring at least one leader survives. In the rare case that multiple leaders remain, anomaly signalling resolves the conflict in expected $O(n)$ time (Section~\ref{leaderelection}).

{\bf Weight System} Let $A$ and $B$ be subpopulations of $n$ agents. The \emph{majority} problem, a predicate from Presburger Arithmetic, is to determine whether $|A|>|B|$, $|A|<|B|$, or $|A|=|B|$.
For efficient majority computation, Angluin et al.~\cite{AngluinAE08} introduced an \emph{approximate majority} protocol: if the initial bias $|A|-|B| = \omega(\sqrt{n}\log n)$, it converges to the majority opinion in $O(\log n)$ time whp. Condon et al.~\cite{condon2020approximate} improved the required bias to $\Omega(\sqrt{n\log n})$.
For the {\em exact majority}, Alistarh et al.~\cite{alistarh2015fastExactMajority} gave the first stable polylogarithmic-time protocol using $\Omega(n)$ states. Subsequent work~\cite{DBLP:conf/soda/AlistarhAEGR17,DBLP:conf/soda/AlistarhAG18,DBLP:conf/podc/BilkeCER17,MocquardAABS2015,mocquard2016optimal,Berenbrink_2020,ben2020log3}
improved the space-time trade-off, culminating in an optimal $O(\log n)$-time protocol with $O(\log n)$ states~\cite{DBLP:conf/focs/DotyEGSUS21}. This matches the $\Omega(\log n)$ state lower bound of Alistarh et al.~\cite{DBLP:conf/soda/AlistarhAG18} for stable sublinear-time protocols under \emph{monotonicity} and \emph{output dominance}, both satisfied by~\cite{DBLP:conf/focs/DotyEGSUS21}.
A related line of research on majority computation within graphical population protocols is presented in~\cite{Ryb1,Ryb2}.
Non-stable protocols (allowing positive error probability) achieve faster convergence whp with less memory. Berenbrink et al.~\cite{Berenbrink_2020} gave a protocol with $O(\log \log n)$ states and $O(\log^2 n)$ time. Also, Kosowski and Uznański~\cite{DBLP:conf/podc/KosowskiU18} achieved $O(1)$ states and $O(\mathrm{polylog}\, n)$ time whp.

We use exact majority computation as a conceptual {\em balance scale} based on {\em system of weights} formed of disjoint subpopulations $W_i$ of agents, where $|W_i|=2^i,$ for all $i=0,\dots,\log(\frac n6).$  
This enables comparison of sub-population cardinalities through appropriately chosen weight combinations. 
Like other majority protocols, ours also uses \emph{cancelling} (opposite biased agents $\to$ unbiased) and \emph{splitting} (biased $\to$ two biased), doubling the bias per phase until $n$. 
In the rare event that the needed set of weights is not created, the corresponding anomaly is signalled (Section~\ref{majority}).

{\bf Remarks} To streamline the presentation, we assume that each subpopulation (input, temporary, or output) has size at most $\frac{n}{6}$; our protocols remain correct and efficient for larger fractions.
Our protocols are non-uniform; namely, we assume that a constant-factor upper approximation $LOG$ of $\log n$ is given and built into the transition function.
Constructing uniform protocols would introduce additional technical difficulties that we deliberately avoid in this work.
Finally, recent work~\cite{DBLP:journals/jcss/CzernerGHE24} gives succinctly encoded, polynomial-time population protocols for Presburger arithmetic, including modular congruences. However, since~\cite{DBLP:journals/jcss/CzernerGHE24} does not address silence, we provide also a polynomial-time silent congruence protocol.

\section{The Parity Problem}\label{s:parity}

Recall that in the {\em parity problem}, agents are tasked with computing the parity of a given sub-population size. In this work, we present a stable population protocol that solves the parity computation problem in polylogarithmic time. 
We begin by presenting a fixed-state protocol that computes the parity of a given subpopulation $X$.
This protocol serves as the fallback mechanism and stabilises in $O(n \log n)$ time, both in expectation and whp.

\subsection{Slow Stable Parity Protocol}
\label{slow-parity}
The slow protocol builds on the leader election protocol of~\cite{AAD+06}. It uses four states: ${L_0, L_1, F_0, F_1}$, with $L_x$ denoting \emph{leaders} and $F_x$ \emph{followers}.
Initially, all agents in subpopulation $X$ start in $L_1$; all others start in $F_0$.
In the interaction of two leaders, the rule
$L_x+L_y\rightarrow L_{(x+y)\bmod 2}+F_0$
applies, where one remains a leader (with updated parity), and the other becomes a follower.
Eventually, a single leader remains, encoding the parity of $|X|$ in its state index. This is then propagated via
rule $L_x + F_y\rightarrow L_x + F_x$, 
so all followers adopt the leader’s index.

\begin{lemma}\label{l:sparity}
    The slow protocol computes parity and stabilises silently in expected time $O(n\log n)$.
\end{lemma}

\begin{proof}
    The leader election protocol, obtained by removing indices from the parity protocol, is stable and stabilises in expected $O(n)$ time~\cite{AAD+06}. The parity of $|X|$ is encoded in the final leader’s state index: $x \equiv |X| \pmod{2}$.
By a Coupon Collector argument, propagating this information to all agents takes expected $O(n \log n)$ time. Thus, the full protocol runs in expected $O(n \log n)$ time.
\end{proof}

\subsection{Efficient Parity Protocol}

An outline of the first efficient, stable, and silent parity protocol (Algorithm~\ref{alg:parity}) follows.

\begin{algorithm}[H]
\label{parity-outline}
\begin{algorithmic}[1]
\Require{~$X$, a set of agents}
\Ensure{\textsc{True} if $|X|$ is even, otherwise \textsc{False}}
\Function{Parity}{$X$}
\State $\Call{LeaderElection}{}$
\State $\Call{WeightCreation}{X}$
\State $Y \gets \emptyset$
\For{$\ell\gets LOG$ \textbf{downto} $1$} 
  \State{\Call{AddWeight}{$\ell$} \Comment{Add weight $2^\ell$ to $Y$}} 
  \State{Compare $X$ and $Y$ using the $\Call{Majority}{}$ protocol}
  \LineIf{$|X|<|Y|$}{$\Call{RemoveWeight}{\ell}$\Comment{Remove weight $2^\ell$ from $Y$}}
\EndFor
\LineIfElse{$|X|=|Y|$}{\Return \textsc{True}} {\Return \textsc{False}}
\EndFunction
\end{algorithmic}
\caption{$\Call{Parity}{}$ Protocol}\label{alg:parity}
\end{algorithm}

{\bf Balance scale} The \emph{balance scale} idea stems from recomputation of powers of 2, 
as proposed in~\cite{AngluinAE08}. 
Equipped with polylogarithmic memory we can afford initial construction of subpopulations $W_i$ of cardinality $2^i,$ for all $i=0,1,2,\dots$. These subpopulations serve as weights in balance scales 
implemented via efficient majority protocol.
Using subpopulation $W_i$, we can determine whether $|X| - |Y| \ge 2^i$ by computing exact majority (with tie detection) between $|X|$ and $|Y|+ |W_i|$. 
This balancing mechanism enables efficient computation of parity and size of any subpopulation.

Implementing this idea as a Monte Carlo protocol is rather straightforward; the challenge lies in designing mechanisms to ensure stable and silent convergence to the correct solution.
This paper implements the weight-based parity computation via a Monte Carlo protocol using $\Call{LeaderElection}{}$, $\Call{WeightCreation}{}$, and $\Call{Majority}{}$ subprotocols. 
The $\Call{LeaderElection}{}$ protocol is efficient and is required for the correct execution of the other two subprotocols.
The $\Call{WeightCreation}{}$ protocol creates weights $W_i$ and the $\Call{Majority}{}$ decides whether $|X|>|Y|,|X|<|Y|$ or $|X|=|Y|$, for any subpopulations $X$ and $Y$.
Our efficient protocol either stabilises silently on the correct parity or signals an anomaly under unusual behaviour. In the latter case, it invalidates its result and triggers the slow parity protocol.
In Section~\ref{sec:mc+}, we develop the tools required to prove the following theorem.

\begin{theorem}\label{th:parity}
There exists an efficient stable population protocol which solves the parity problem using $O(\log^3 n)$ states and stabilises silently both whp and in expected $O(\log^3 n)$ time.
\end{theorem}

The proof of Theorem~\ref{th:parity} is provided in Subsection~\ref{s:fullparity}.

\section{Efficient Protocols with Anomaly Signalling}\label{sec:mc+}

In this section, building on prior work on efficient stable protocols, we develop a general method for constructing efficient population protocols that stabilise silently with certainty.
Our approach is based on the MC+ paradigm. 
Instead of handling failures of fast protocols in an ad hoc manner, we use systematic anomaly signalling to detect unusual execution patterns that may compromise the correctness of the fast computation.
This allows several efficient subprotocols to be composed sequentially, with the output of one serving as the input of the next. If an anomaly is signalled at any stage, the computation switches to the slower always-correct backup protocol.
Thus, the efficient part runs in $O(\mathrm{polylog}\,n)$ time and space and produces the correct output with high probability. Otherwise, it signals an anomaly and falls back to guaranteed correctness. An anomaly is treated as a failure of the efficient computation even when the provisional result happens to be correct.

\subsection{Clock Protocol}\label{s:clock}

Central to our design is the clock protocol, which acts as a clocking device and a command counter, tracking algorithmic progress and synchronising all agents to the same point in time. In rare cases of desynchronisation, it signals an anomaly.
Each main (e.g., parity) protocol is built by sequentially composing subprotocols, each spanning one or more rounds.
The clock counts down $R = O(\mbox{polylog}\,n)$ rounds of the main protocol, each taking $O(\log n)$ time. To measure a round’s duration, a leader-maintained interaction counter tracks $c \log n$ interactions, see Lemma~\ref{l:clock}, triggering the transition to the next round.
As throughout the paper, we favour simple clock mechanisms over optimal parameter tuning to keep the presentation clear.

{\bf Components of the clock protocol}
\begin{itemize}
    \item Each agent maintains a round number $r$.
    \item Leaders maintain an auxiliary counter $s$ to count interactions within a round. After $\Theta(LOG)$ interactions, they increment $r$ by 1, advancing to the next round until $r = R$.
    \item If two agents with round numbers $r$ and $r-1$ interact, the agent with $r-1$ adopts $r$.
    \item If two agents' round numbers differ by more than 1, a clock anomaly is signalled.
\end{itemize}

\begin{lemma}\label{l:clock}
For any constant $C > 0$, there exists a constant $c > 0$ such that counting $c \log n$ interactions takes at least $C \log n$ time whp.
Moreover, there exists a constant $C_0 > 0$ such that, within time $C_0 \log n$ after the first leader enters round $r$, all other agents advance to round $r$ whp. 
\end{lemma}

\begin{proof}
Each agent experiences in expectation $2C \log n$ interactions in time $C \log n$. By the Chernoff bound, this number does not exceed $c \log n$ whp, for sufficiently large $c$. The epidemic time to propagate round $r$ is at most $C_0 \log n$ whp, for some $C_0 > 0$, ensuring all agents adopt round $r$ within this time.
\end{proof}

In the clock protocol, we choose $c > 0$ such that counting $c \log n$ interactions takes time $C \log n$ whp, ensuring that the interval $(C - C_0) \log n$ suffices for all round-$r$ computations whp. By Lemma~\ref{l:clock}, all agents occupy round $r$ simultaneously for at least $(C - C_0) \log n$ time whp. Also by Lemma~\ref{l:clock}, round numbers differ by at most 1 across any two agents throughout all $R$ rounds whp. Thus, no clock anomaly is signalled whp.

\begin{lemma}\label{l:clocktime}
    The expected time for the clock protocol to terminate after $R$ rounds, either by completion or anomaly signalling, is $O(R \log n)$.
\end{lemma}

\begin{proof}
Let us consider a modified clock protocol in which anomaly signalling is replaced by a new rule.
This rule specifies that when two agents interact in rounds $r$ and $r'$, with $r' < r$, the agent in round $r'$ adopts round $r$. Observe that the anomaly-signalling clock protocol terminates in time no worse than the modified one. The two protocols either complete their counting during the same time, or, if an anomaly is signalled, the original clock finishes earlier than the modified one.  To bound the time for the modified clock protocol to terminate after $R$ rounds, consider the computation of an agent that remains a leader throughout.
For the modified clock to finish, firstly, this agent must count $R$ rounds. 
Counting each of them takes no more than $c\log n$ interactions of the leader.
Hence, counting $R$ rounds takes the leader expected time $O(R\log n)$. 
Secondly, all the remaining agents must also advance to round $R$, which occurs in time no worse than the time of the epidemic process by which the leader, already in round $R$, spreads round $R$ to the population. 
This takes the expected time $O(\log n)$. 
Therefore, the modified clock protocol terminates in expected time $O(R\log n)$. As noted earlier, this also implies that the clock protocol with anomaly signalling terminates in expected time $O(R\log n)$.
\end{proof}

\subsection{MC+ Protocols}\label{s:MC+}

In this subsection, we define MC+ protocols, which compute the correct result with high probability. 
The execution of an MC+ protocol is restricted to a prescribed polylogarithmic number of rounds counted by the clock. 
Unlike standard Monte Carlo protocols, however, they also provide a conditional guarantee of correct task execution, i.e., if no anomaly is signalled, then the computed result is correct. 
We distinguish two types of anomalies:
\begin{description}
    \item[A1] \emph{clock anomaly}, arises when agents whose rounds differ by more than 1 interact,
    \item[A2] \emph{protocol anomaly}, involves the signalling of irregularities specific to the protocol itself.
\end{description}

In most cases, signalling protocol-specific anomalies only requires agents to transition between successive rounds. Therefore, when clock anomalies occur, it is typical to lose the data needed to signal protocol-specific anomalies. Another argument for specifically mentioning clock anomalies is the difficulty of assigning them to a specific protocol when multiple protocols are combined.

Once one MC+ protocol has completed, its output may be used by the next MC+ protocol in a sequential composition, provided that the preceding protocol has executed correctly. If any component executes incorrectly, the resulting inconsistency eventually triggers an anomaly, causing the computation to switch to a stable backup protocol.
For example, the leader-election protocol is designed so that all leaders can never disappear, as this would cause the clock to stop. Nevertheless, with negligible probability, several leaders may remain after the protocol has completed. Subsequent protocols, which operate under the assumption that a unique leader has been elected, may then fail to produce a meaningful result. However, two distinct leaders eventually interact with probability~1, causing an anomaly to be signalled and the computation to switch to the backup protocol.

We define also a \emph{membrane configuration} as the multiset of agent states at the boundary between two consecutive rounds $r$ and $r+1$, as determined by the clock. It may be viewed either as the configuration immediately before leaving round $r$ or as the configuration immediately after entering round $r+1$.
The \emph{input configuration} of an MC+ protocol is the membrane configuration upon entering its first round, whereas its \emph{output configuration} is the membrane configuration upon exiting its last round. Thus, in a sequential composition of MC+ protocols, the output configuration of one protocol serves as the input one of the next.

{\bf State space}
This part demonstrates how to effectively integrate multiple population protocols, necessitating a more sophisticated state space structure.
In particular, agent states are recorded as tuples $(A_0)(A_1)(A_2)(A_3)$. 
Element $A_0$ contains the input of the protocol, i.e.,\ whether the agent belongs to set $X$.
Element $A_1$ contains global information computed by the protocol, i.e.,\ whether the agent is a leader or belongs to one of the weights and whether this weight is currently placed on the scale pan.
Element $A_1$ will also contain the output.
We call this element the {\em global state} of an agent.
Element $A_2$ specifies the round number $r$ and contains the leader's interaction counter in a given round. 
Element $A_3$ contains local information specific to the subprotocol executed in round $r$. We call this part the {\em local state} of an agent.
The transition $(A_0)(A_1)(A_2)(A_3)+(B_0)(B_1)(B_2)(B_3)\rightarrow (A_0')(A_1')(A_2')(A_3')+(B_0')(B_1')(B_2')(B_3')$ is implemented as follows. First, we use the clock transition rule to calculate the round number in states $A_2'$ and $B_2'$, optionally incrementing the interaction counter(s) in the round. Simultaneously, we calculate other states $A_3',B_3'$ according to the transition rules for the subprotocol. If necessary, we also update the information whether the agent is still the leader or about the weights or output in $A_1',B_1'$. Elements $A_0,B_0$ remain unchanged.

{\bf MC+ protocol} To capture the desired properties, we define that an MC+ protocol:
\begin{description}
    \item[C1] 
Is active for a specified number of rounds, meaning that its operations are restricted to $R=O(\mbox{polylog}\,n)$ consecutive rounds. 
If no anomaly is signalled during or after the protocol's execution, the output configuration reached after its final round is silent.
    \item[C2] 
    There is a negligible probability that a protocol anomaly signal will be triggered during the protocol's execution or within an expected polynomial time less than $n^a$ following its completion. Such an anomaly signal would stem from a state inconsistency, potentially indicating erroneous behaviour in the protocol.
    \item[C3] 
    The absence of such a signal, together with the lack of clock anomaly, indicates a successful execution of the task.
\end{description}

The actions of each MC+ protocol are restricted to a prescribed number of rounds. Consequently, if no anomaly occurs, then after the protocol exits its last round, no meaningful transition belonging to that protocol remains enabled, and the sequential composition reaches a silent configuration.
Thus, every MC+ protocol is \emph{conditionally} stable and silent; that is, it satisfies both properties provided that no anomaly is signalled.

\begin{theorem}\label{th:stable+mc+}
If an efficient MC+ protocol and a polynomial-time stable protocol exist for the same problem, they can be combined to yield a protocol that is both efficient and stable.
Moreover, if the polynomial-time stable protocol is silent, the resulting combined protocol is also silent.

\end{theorem}

\begin{proof}
Assume that the efficient protocol completes its computations in $R$ rounds.
If not successful, it reports a protocol anomaly in at most $n^a$
expected time with negligible probability less than $n^{-a}$, for suitably chosen protocol parameters and $a>0$.
In this time, by Lemma~\ref{l:clock} the clock reports an anomaly with negligible probability less than $n^{-a}$.
Assume also that the slow algorithm runs in expected time less than $n^a$.
The expected execution time of the combined algorithms consists of the following components:
\begin{itemize}
\item the expected clock completion time which by Lemma \ref{l:clocktime} is $O(R\log n)$,
\item the execution time of the slow protocol times the probability of clock anomaly: $n^{a}\cdot n^{-a}=1$, 
\item the expected time needed to handle protocol anomaly times its probability: $(n^a+n^a)\cdot n^{-a}=2$.
\end{itemize}

Therefore, the expected execution time of the combination of efficient and slow protocol is at most
$ O(R\log n)+1+2=O(R\log n).
$
Moreover, the combination of the protocols is stable and silent.

\end{proof}

\begin{theorem}\label{th:combine_mc+}
A sequence of a polylogarithmic number of MC+ protocols is an MC+ protocol.
\end{theorem}

\begin{proof}
Assume that there is a sequence of $k=\mbox{polylog}\,n$ MC+ protocols that are executed in rounds $R_1,R_2,\ldots, R_k$ respectively.
There exists a constant $a>0,$ s.t., for every protocol the probability of signalling a protocol anomaly is at most $n^{-a}$ and the expected time needed to signal an anomaly is less than $n^a$.
Consequently, the probability that any protocol in the sequence signals an anomaly is at most $kn^{-a}$, which is negligible.
The total number of rounds is $R=\sum_i R_i$, which remains polylogarithmic.
Moreover, the expected time needed to signal a protocol anomaly after completion of the last protocol is at most~$n^a$.
\end{proof}

\subsection{MC+ Leader Election}\label{leaderelection}

The MC+ protocol for leader election proposed in this work builds upon a scheme that has been employed in several prior studies~\cite{AngluinAE08, DBLP:journals/jacm/GasieniecS21, DBLP:conf/stoc/BerenbrinkGK20}.
In this approach, the leaders perform coin flips in successive rounds.
Whether an agent draws heads or tails depends on the random scheduler’s designation of that agent as either the initiator or the responder in an interaction.
If, in a given round, one leader flips tails while another flips heads, the leader obtaining tails becomes a follower.
Unlike previous protocols, our method associates the coin-flip outcome (heads or tails) with the round number, thereby ensuring that at least one leader always persists.
Although the protocol is less efficient than several leader election protocols known from the literature, this construction reflects our emphasis on clarity and conceptual simplicity.
Adapting more sophisticated existing protocols to the MC+ framework would likely entail additional technical difficulties.
Furthermore, leader election does not constitute the bottleneck for the efficiency of parity or congruence protocols.

Besides detectable clock anomalies, the only other potential failure is persistent multiple leaders, which is signalled as a protocol anomaly upon interaction after leader election concludes. This signalling occurs in expected time $O(n)$.

\begin{algorithm}[htb!]
\begin{algorithmic}[1]\label{a:2}
\Require{A subpopulation of agents}
\Ensure{A subpopulation of agents with a single leader}
\Function{LeaderElection}{}

\State{All agents adopt state $L$, indicating that each is a leader (candidate)}
\For{round $r\gets 0$ \textbf{to}  $d\cdot LOG$ }

\State{Each leader flips a symmetric coin obtaining \texttt{H} or \texttt{T}}

\DLineIf{an agent flipped \texttt{H} as a leader or learned that \texttt{H} was flipped in $r$}{It spreads by epidemic information that \texttt{H} was flipped in $r$}

\DLineIf{a leader flipped \texttt{T} and learned that \texttt{H} was flipped in $r$}{This leader changes type to a follower}

\EndFor

\LineIf{after protocol end there is an interaction between two leaders}{\textbf{signal}\ \textsc{Anomaly}}

\EndFunction
\end{algorithmic}
\caption{$\Call{LeaderElection}{}$ Protocol}\label{alg:leader}
\end{algorithm}

\begin{theorem}
There exists a constant $d>0$ such that  $\Call{LeaderElection}{}$ (Algorithm~\ref{alg:leader}) is an MC+ protocol,
which either leaves a single leader in the output configuration whp or reports an anomaly.
Moreover, protocol anomaly is signalled in the expected time $O(n)$ following the completion of the protocol.\label{th:leader}
\end{theorem}

\begin{proof}
In the input configuration, all agents are leaders and the goal of the protocol is to leave exactly one leader in the output configuration.
Let $A_r$ denote an event in round $r$, defined as follows.
If round $r$ has at least two leaders, $A_r$ is the event that at least half, but not all, of the leaders flip tails.
In this case, the probability of $A_r$ is at least $3/8$ (with the minimum occurring when there are three leaders).
When event $A_r$ occurs, since $c$ is chosen to be sufficiently large,
the epidemic in round $r$ succeeds whp, ensuring that at least half of the leaders are eliminated.
If round $r$ has exactly one leader, $A_r$ is defined as the event that an independent draw with success probability $3/8$ succeeds.

A sufficient condition for electing a single leader is that event $A_r$ occurs in at least $\log n$ rounds and all epidemics spreading the information about the heads toss are successful.
The second requirement is fulfilled whp.
Over $d\cdot LOG$ rounds, the expected number of occurrences of $A_r$ is at least $\frac{3}{8} d \log n$.
By the Chernoff bound, for sufficiently large $d$, event $A_r$ occurs at least $\log n$ times whp.
Hence, a single leader is elected whp.
If not, an anomaly is signalled in the expected time $O(n)$ required for two remaining leaders to interact. 
\end{proof}

\subsection{MC+ Majority}\label{majority}

In this section, we construct and analyse our MC+ exact majority computation protocol.  
We conjecture that the best known majority protocol, achieving time $O(\log n)$, can be adapted into the MC+ protocol operating within a single or fixed number of rounds. Incorporating such a protocol into parity computation would reduce its overall time to $O(\log^2 n)$. However, owing to the complexity of this task, we do not implement it here and leave it as an open problem.
The three possible outcomes of the protocol, $|X|>|Y|$, $|X|<|Y|$, and $|X|=|Y|$, are represented by all agents assuming the states $X_{win}$, $Y_{win}$, and $tie$, respectively, in the output configuration.

\begin{algorithm}[htb]
\begin{algorithmic}[1]\label{a:maj}
\Require{Subpopulations of agents $X$ and $Y$}
\Ensure{Determine whether $|X|>|Y|$, $|X|<|Y|$, or $|Y|=|X|$.}
\Function{Majority}{$X, Y$}

\State{Assign state $(1)$ to agents in $X \setminus Y$, $(-1)$ to those in $Y \setminus X$, and $(0)$ to the rest.}

\For{round $r\gets 0$ \textbf{to}  $LOG+4$} 

  \LineIf{interacting agents are in the same round $r$}
 execute when applicable:
  \State{\hspace{\algorithmicindent} $(a) + (b) \rightarrow \left( \lfloor \frac{a+b}{2} \rfloor \right) + \left( \lceil \frac{a+b}{2} \rceil \right)$ or} 
  \State{\hspace{\algorithmicindent} $(x_{win})+(0|1|2)\rightarrow (x_{win})+(x_{win})$ or} \State{\hspace{\algorithmicindent} $(y_{win})+(-2|-1|0)\rightarrow (y_{win})+(y_{win})$.}

  \State{{\bf upon} leader entering round $r$ in state $(a)$~{\bf do}}
   \indLineIfEIEIEI{$r<LOG+4$ \textbf{and} $(a)=(-1|0|1)$}
   {$(a)\rightarrow(2a)$}
   {$r<LOG+4$ \textbf{and} $(a)=(-2|2)$}
   {$(-2)\rightarrow (y_{win})$ or $(2)\rightarrow (x_{win})$}
  {$(a)=(x_{win}|y_{win})$}{
  $(x_{win})\rightarrow (X_{win})$ or $(y_{win})\rightarrow (Y_{win})$}
  {$r=LOG+4$ \textbf{and} $(a)=(-2|-1|0|1|2)$}{$(a)\rightarrow(tie)$}

    \State{{\bf upon} follower entering round $r$ in state $(a),$ triggered by an agent in state $(b)$~{\bf do}}
    
    \indLineIf{$(b) \neq (X_{win}|Y_{win}|tie)$} 
 
   {\indindLineIfElseElse{$(a)=(-1|0|1)$}
   {$(a)\rightarrow(2a)$}
   {$(a)=(-2|2)$}
   {$(-2)\rightarrow (y_{win})$ or $(2)\rightarrow (x_{win})$}
   {\textbf{signal} \textsc{Anomaly}}}

   \indLineIf{$(b)=(X_{win})$}
   {\indindLineIfElse{$(a)=(x_{win}|X_{win})$}{$(x_{win})\rightarrow (X_{win})$}
   {\textbf{signal} \textsc{Anomaly}}}

   \indLineIf{$(b)=(Y_{win})$}
   {\indindLineIfElse{$(a)=(y_{win}|Y_{win})$}{$(y_{win})\rightarrow (Y_{win})$}
   {\textbf{signal} \textsc{Anomaly}}}

   \indLineIf{$(b)=(tie)$}
   {\indindLineIfElse{$(a)=(-2|-1|0|1|2)$}{$(a)\rightarrow (tie)$}
      {\textbf{signal} \textsc{Anomaly}}}

\EndFor

\EndFunction
\end{algorithmic}
\caption{$\Call{Majority}{}$ Protocol}\label{alg:majority}
\end{algorithm}

Our MC+ majority protocol (see Algorithm~\ref{alg:majority}) operates over $LOG + 4$ rounds and achieves a time complexity of $O(\log^2 n)$. 
Throughout this construction, we assume that a unique leader has already been selected.
The protocol compares two subpopulations, $X$ and $Y$, determining which is larger or whether they are of equal size.
Notably, $X$ and $Y$ need not be disjoint.
At the start of the first round, agents are assigned states as follows:
elements exclusive to $X$ are assigned state $(1)$;
elements exclusive to $Y$ state $(-1)$;
elements in both 
$X$ and $Y$ 
or in neither state $(0)$.
During each round $r$, an averaging operation is performed on pairs of agents in round $r$:
$(a) + (b) \rightarrow \left( \lfloor \frac{a+b}{2} \rfloor \right) + \left( \lceil \frac{a+b}{2} \rceil \right)$.
When transitioning to the next round, a doubling operation is applied: $(a) \rightarrow (2a)$.

If, when moving to round $r$, we attempt to perform a doubling operation on the state $(a)=(-2|2)$ (i.e., $(a)=(-2)$ or $(a)=(2)$), the following transition is performed: $(2)\rightarrow(x_{win})$ or $(-2)\rightarrow(y_{win})$.
Once agents in state $(x_{win})$ or $(y_{win})$ appear, the corresponding state is propagated by the epidemic process
on states not contradicting this outcome: $(x_{win})+(0|1|2)\rightarrow (x_{win})+(x_{win})$ or 
$(y_{win})+(-2|-1|0)\rightarrow (y_{win})+(y_{win})$.
Note that, e.g.,\ states $(-1|-2)$ are not susceptible to infection by state $(x_{win})$.
In this scenario, when moving on to round $r+1$, all agents are in state $(x_{win})$ or $(y_{win})$ whp.
The leader is the first to advance to round $r+1$ and makes whp the transition:
$(x_{win})\rightarrow (X_{win})$ or $(y_{win})\rightarrow (Y_{win})$.
Subsequently, all other agents with state $(x_{win})$ or $(y_{win})$ moving to round $r+1$ perform exactly the same transition as the leader whp.
The state $(X_{win})$ or $(Y_{win})$ remains unchanged until the final round of the majority protocol, unless an anomaly is signalled.
This state indicates that the inequality $|X|>|Y|$ or $|X|<|Y|$ holds, respectively.

There are two unlikely cases when an anomaly is signalled.
The first such case is when one of the agents has a state different from the leader's state $(x_{win})$ or $(y_{win})$
when moving to round $r+1$.
The second case is that the leader had a state other than $(x_{win})$ or $(y_{win})$ when entering round $r+1$, while another agent entering that round has one of these states.

If $X$ and $Y$ have the same cardinality,
the protocol avoids states $(x_{win})$ and $(y_{win})$ whp.
If, at the end of round $LOG+3$, the leader is in state $(-2|-1|0|1|2)$, it adopts state $(tie)$ in the next round.
This state $(tie)$ is then propagated by the epidemic process to all agents entering round $LOG+4$ in states different from $(x_{win})$ and $(y_{win})$.
However, if an agent enters round $LOG + 4$ in state $(x_{win})$ or $(y_{win})$, a protocol anomaly is signalled (as already specified in the previous paragraph).

In our analysis of the majority protocol, we use a theorem from~\cite{berenbrink2018simpleloadbalancing}. We restate it below in a form adapted to our model.

\begin{theorem}[Theorem 1 in \cite{berenbrink2018simpleloadbalancing}]\label{load_balancing}
Suppose we are given an input configuration in which each agent has a state $(a)$, where $a$ is an integer and the discrepancy between states is $\delta$. A protocol using only averaging arrives at a configuration with a discrepancy between states of at most $2$ after time $O(\log \delta + \log n)$~whp.
\end{theorem}

\begin{theorem}\label{majority_mc+}
Our majority protocol signals an anomaly with negligible probability.
In the absence of an anomaly, the protocol's output is always correct.    
\end{theorem}

\begin{proof}
Let bias $\Delta=|X|-|Y|$ and ${\mathcal S}_r$ be the sum of states $(a)$ in which each agent starts round $r$.
Note that ${\mathcal S}_0=\Delta$.
The averaging operation does not change the sum of agent states $(a)$ in round $r$.
Thus as long as there are no $(x_{win})$ or $(y_{win})$ states,
${\mathcal S}_r$ is also the sum of states $(a)$ in which each agent ends round $r$.
Due to doubling operation ${\mathcal S}_{r+1}=2{\mathcal S}_r$, as long as states $(x_{win})$ and $(y_{win})$ are not present.    

Consider two cases. 
In the first case, one of the sets $X$ or $Y$ is larger than the other. 
Without loss of generality, assume this set is $X$. 
If there are no $(x_{win})$ or $(y_{win})$ states in round $r$, then ${\cal S}_r=2^{r-1}\Delta\ge 2^{r-1}$ and 
${\cal S}_r\le 2n$. 
If no $(x_{win})$ or $(y_{win})$ states had appeared until round $LOG+3$, then the sum ${\cal S}_{LOG+3}$ would be at least $2^{LOG+2}$ which is greater than $2n$.
This proves that states $(x_{win})$ or $(y_{win})$ appear at the latest in round $LOG+3$.

We will now describe the dominant scenario that occurs whp.
At the end of each round $r$, before the state $(x_{win})$ appears,
the discrepancy between the values $(a)$ is at most $2$ whp due to Theorem \ref{load_balancing}.
Furthermore, since ${\cal S}_r>0,$ some values $(a)$ must be positive.
If eventually all values are $(a)=(-1|0|1)$, then they will double during transfer to round $r+1$.
If eventually all values are $(a)=(0|1|2)$ and some value $(a)=(2)$, then the state $(x_{win})$ is created in round $r+1$.
As we have already mentioned, the state $(x_{win})$ must appear in round $r+1\le LOG+3$.
In round $r+1$, the state $x_{win}$ is propagated to all agents in the population through the epidemic.
This causes all agents in round $LOG+4$ to adopt state $(X_{win})$.

Now, let's assume that no clock or protocol anomaly was signalled.
As mentioned before, in some round $r\le LOG+3$, state $(x_{win})$ or $(y_{win})$ must have occurred.
So the conclusion $|X|=|Y|$ cannot be reached.
Since no anomaly was signalled, at the end of round $r$, the leader must have also adopted the same state $(x_{win})$ or $(y_{win})$, as well as every other agent.
Therefore, at the beginning of round $r$, all agents must have had state $(x_{win})$ or state $(0|1|2)$. 
Since $\Delta>0$, all states in any round cannot be equal to $(y_{win})$ or $(-2|-1|0)$.
Thus, the outcome $|X|<|Y|$ cannot be reached.
This~guarantees that all agents adopt the correct final state $(X_{win})$ meaning $|X|>|Y|$.

In the second case, bias $\Delta=0$.
In the dominant scenario, which holds whp, at the end of each round $r<LOG+4$, there are three possible states $(a)=(-1|0|1)$ due to Theorem \ref{load_balancing}.
Upon entering round $LOG+4$, the leader adopts $(tie)$, followed by all other agents.

Now, let's assume that no anomaly was signalled.
Assume also by contradiction that in some round $r<LOG+4$, state $(x_{win})$ or $(y_{win})$ appeared -- without loss of generality, let this be $(x_{win})$.
Then, since $\Delta=0$, state $(y_{win})$ or $(-2|-1)$ must also have appeared in round $r$.
However, this would result in the signalling of an anomaly in round $r+1$.
This proves that state $(x_{win})$ or $(y_{win})$ did not appear before round $LOG+4$.
This means that when moving to round $LOG+4$, both the leader and all other agents adopt the state $(tie)$, which reflects the correct outcome $|X|=|Y|$.
\end{proof}

\subsection{System of Weights}\label{weights}

We introduce here the MC+ protocol, which creates \emph{weights} that are disjoint sets of agents.
Weight $W_i$ is a set consisting of $2^i$ agents. 
The weight creation rounds are indexed by integers $r=0,1,2,\ldots,LOG + 2$. 
In round $r$, the protocol creates $W_r$, consisting of $2^r$ agents.
If each of these agents is assigned unit mass, then  $W_r$, composed of $2^r$ agents, has a total mass of $2^r$.

Once created, weights can be added to or removed from a set $Y$ using the {\sc AddWeight} and {\sc RemoveWeight} operations. These operations are performed via an epidemic starting from the leader, which takes one round. The absence of clock anomalies guarantees the correct execution of these operations. Note that the creation of the set $Y=\bigcup_k W_k$ at the end of {\sc WeightCreation} can be done in a similar way.

At the start of the {\sc WeightCreation} protocol, all followers are in state $\phi$, which marks free agents.
At the beginning of each subsequent round $r$ the leader possessing state $\lambda$
adopts state $\lambda_r$, which serves as the precursor of weight $W_r$ carrying its mass $2^r$.
The leader then transfers this mass to the first free agent it encounters via transition: $\lambda_r+\phi\rightarrow \lambda+w_r(r)$,
where $\lambda$ denotes a non-final leader state that carries no mass.
Agents participating in the creation of $W_r$ have states $w_r(i)$, indicating that the agent possessing it belongs to $W_r$ and holds mass $2^i$ from that weight.
Creating the weight is accomplished by a transition that divides the mass greater than $2^0=1$ between two agents: $w_r(i)+\phi\rightarrow w_r(i-1)+w_r(i-1)$ when $i>0$.
This process completes when all mass carrying agents have state $w_r(0)$ of mass 1. These agents then change their local state $w_r(0)$ to global state $w_r$ in the output configuration marking agents of weight $W_r$.

In some round $r$, the weight $W_r$ will inevitably fail to unfold.
This will lead to a transition to round $r+1$ with state $\lambda_r$ or $w_r(i)$ for $i>0$ present in the population.
In the protocol, we mark such an event with a transition to state $m(1)$. 
State $m(1)$ then propagates to the entire population via an epidemic within round $r+1$ whp. The leader moving to a new round $r+2$ in state $m(1)$ changes state to $m(2)$. Then, all followers moving to round $r+2$ in state $m(1)$ will receive state $m(2)$. We signal an anomaly if the leader does not receive state $m(1)$ during the epidemic in round $r+1$. This causes an anomaly to be signalled when a follower in state $m(1)$ attempts to transition to round $r+2$. Note that if state $m(1)$ does not appear when transitioning to round $r+1$, the weight $W_r$ can be stored in the global states. This is done when transitioning to round $r+2$.
With negligible probability, too few weights may be created to perform the parity protocol. To detect such an event, after the weights are created, their total mass is compared to the number of elements in the set $X$ for which the parity protocol is being performed. If an insufficient number of weights are created, an anomaly is signalled.

The absence of weights with indices between that of the largest created weight and $LOG$, i.e., these sets are empty, does not hinder the parity protocol.

\begin{lemma}\label{weight_logn}
Assume that at the beginning of round $r$ all 
$W_i,$ for $i\le r-1,$ have been created and the population has at least $2^{r+1}$ agents in state $\phi$. Then $W_r$ will be created in time $O(\log n)$ whp.    
\end{lemma}

\begin{proof}
There are exactly $2^r$ agents that do not have state $\phi$ at the start of round $r$. These agents include the leader and those involved in all weights created up to round $r-1$. The probability that an agent in state $w_r( i)$ encounters an agent in state $\phi$ during an interaction is at least $\frac{2}{n} \cdot \frac{2^{r+1} - 2^r}{2^{r+1} + 2^r} \ge \frac{2}{3n}$.

The process of creating $W_r$ can be modeled as a binary tree, with the root representing state $w_r( r)$. This state has two children, each of mass $2^{r-1}$, corresponding to state $w_r( r-1)$. At the next level, there are four states $w_r( r-2)$, each with mass $2^{r-2}$, and so on. At the final level, there are $2^r < n$ states $w_r( 0)$, each of unit mass, represented as leaves.

Consider a path from the root to a leaf in this tree. This path consists of $r$ edges, each representing a transition $w_r( i) + \phi \rightarrow w_r( i-1) + w_r( i-1)$. We show that all transitions along any fixed path are completed in time $O(\log n)$ whp. The probability of executing the next transition on the path in a given interaction is at least $\frac{2}{3n}$. Over $cn \log n$ interactions, the expected number of interactions between the current agent on the path and an agent in state $\phi$ is at least $\frac{2}{3} c \log n$. By the Chernoff bound, this number exceeds $\log n>r$ whp for sufficiently large $c$. Thus, each transition along any fixed path completes in time $O(\log n)$ whp. By applying the union bound over all paths, we conclude that all transitions across all paths complete in time $O(\log n)$ whp. Consequently, the construction of $W_r$ also completes in time $O(\log n)$ whp.
\end{proof}

\begin{algorithm}[htb]
\begin{algorithmic}[1]\label{a:3}
\Require{A set of agents with a single leader}
\Ensure{Weights created among non-leaders}
\Function{WeightCreation}{$X$}

\For{round $r\gets 0$ \textbf{to}  $LOG+2$} 

  \State{\textbf{upon} entering round $r$ triggered by an agent in state different than $m(2)$:}
\indLineIfElse{$r=0$}{\State{\hspace{\algorithmicindent}\hspace{\algorithmicindent} the leader adopts state $\lambda_0$ and non-leaders adopt state $\phi$.}}
  {\State{\hspace{\algorithmicindent}\hspace{\algorithmicindent} the leader changes state $\lambda$ to state $\lambda_r$ (with mass $2^r$),}
  \State{\hspace{\algorithmicindent}\hspace{\algorithmicindent} a follower changes local state $w_{r-2}(0)$ to global state $w_{r-2}$ and local $w$,}
  \State{\hspace{\algorithmicindent}\hspace{\algorithmicindent} the leader changes state $\lambda_{r-1}$ to state $m(1)$,}
  \State{\hspace{\algorithmicindent}\hspace{\algorithmicindent} a follower changes state $w_{r-1}(i)$ to state $m(1)$, when $i>0$,}  
  \State{\hspace{\algorithmicindent}\hspace{\algorithmicindent} the leader changes state $m(1)$ to state $m(2)$,}
  \State{\hspace{\algorithmicindent}\hspace{\algorithmicindent} a follower in state $m(1)$ \textbf{signals}\ \textsc{Anomaly}.}}  
  
  \State{\textbf{upon} entering round $r$ triggered by agent in state $m(2)$:}
  \State{\hspace{\algorithmicindent} a follower adopts state $m(2)$.}

  \DLineIf{interacting agents are in the same round}{
   \hspace{\widthof{ }}the leader transfers mass $2^r$ to the first free agent it encounters:\\
  \hspace{\algorithmicindent}\hspace{\algorithmicindent}\hspace{\algorithmicindent}\hspace{\widthof{\textbf{if}}}$\lambda_r + \phi \rightarrow \lambda + w_r( r)$,\\
  \hspace{\algorithmicindent}\hspace{\algorithmicindent}\hspace{\algorithmicindent}\hspace{\widthof{\textbf{if}}}distribute the mass by repeating:\\
  \hspace{\algorithmicindent}\hspace{\algorithmicindent}\hspace{\algorithmicindent}\hspace{\widthof{\textbf{if}}}$w_r( i) + \phi \rightarrow w_r( i-1) + w_r( i-1)$, for $i>0$,\\
  \hspace{\algorithmicindent}\hspace{\algorithmicindent}\hspace{\algorithmicindent}\hspace{\widthof{\textbf{if}}}epidemic of state $m(1)$: $m(1) + \text{\emph{any state}} \rightarrow m(1)+m(1)$.}

\EndFor
\State{Compare $\bigcup W_k$ and $X$ using the $\Call{Majority}{}$ protocol}
\State{(elements of $\bigcup W_k$ have global states $w_k$)}
\LineIf{$|\bigcup W_k|<|X|$}{\textbf{signal}\ \textsc{Anomaly}}

\EndFunction
\end{algorithmic}
\caption{$\Call{WeightCreation}{}$ Protocol}\label{alg:weights}
\end{algorithm}

\begin{theorem}\label{weights_mc+}
 $\Call{WeightCreation}{}$ protocol (Algorithm~\ref{alg:weights}) creates $W_0,W_1,W_2,\ldots, W_{Max(r)}$ correctly with high probability, for some value $Max(r)>\log \frac{n}{6}$. 
With negligible probability, it correctly creates $W_0,W_1,W_2,\ldots, W_{Max(r)}$, for some $Max(r)\le\log \frac{n}{6}$.

\end{theorem}

\begin{proof}
By Lemma \ref{weight_logn}, the duration of the round is sufficient to produce $W_r$ whp when $2^r\le n/3$.
This ensures that all weights $r \leq \log \frac{n}{3}$ are created correctly whp.
With negligible probability, the first incompletely formed $W_r$ satisfies $r \leq \log \frac{n}{3}$.

Let $W_r$ be the first weight that is not fully formed in its round.
This occurs when, upon transitioning to round $r+1$, the state $\lambda_r$ or $w_r(i),$ for some $i > 0,$ persists.
Consequently, the state $m(1)$ emerges upon entering round $r+1$.
This state is propagated whp to all agents in the population via an epidemic process during round $r+1$.
When the leader advances to round $r+2$, it updates the state $m(1)$ to $m(2)$. 
This causes all other agents to adopt state $m(2)$ at the start of the round $r+2$ whp.
With negligible probability, state $m(1)$ fails to reach the leader in round $r+1$.
Assume also, no clock anomaly occurs.
Then, the leader may advance to round $r+2$ without establishing the state $m(2)$.
This, in turn, triggers a protocol anomaly when a follower in state $(1)$ progresses to round $r+2$.
\end{proof}

\subsection{Efficient Stable Parity Protocol -- Proof of Theorem \ref{th:parity}}
\label{s:fullparity}

We provide here a proof of Theorem \ref{th:parity}, 
showing that there exists efficient stable protocol for the parity problem using $O(\log^3 n)$ states and silently stabilising in $O(\log^3 n)$ time.

\begin{proof}[Proof of Theorem \ref{th:parity}]
 We can use the scheme described in $\Call{Parity}{}$ protocol (Algorithm~\ref{alg:parity}) to combine the MC+ protocols for leader election, majority, and weight creation, obtaining an efficient Monte Carlo protocol for parity.
 By Theorem \ref{th:combine_mc+}, this combination is proved to be an MC+ protocol for parity.
 In order to implement this combination, the global state component contains a bit that specifies whether the agent belongs to a weight that contributes to set $Y$.
 This bit is turned on or off by the result of the majority operation stored in its final local state.
 The number of states used by this process is $O(\log^3 n)$.
 This can be shown as follows.
 The leader in each of the MC+ protocols needs a fixed number of states except for the clock component, which requires $O(\log^3 n)$ states.
 Non-leaders need $O(\log^2 n)$ states in the clock component. Those that belong to the weights use $O(\log n)$ states in the global component to determine the weight number and a constant number of states in the local component. Non-leaders that are not in the weights have a constant number of states in the global component and a logarithmic number in the local component.

 By Theorem \ref{th:stable+mc+}, the fast MC+ parity protocol can be combined with a slow parity protocol, resulting in an efficient stable parity protocol.
 This protocol is silent because the slow parity protocol is silent.
 The total number of rounds of the MC+ parity protocol is $R=O(\log^2 n),$ so it successfully computes the parity whp in time $O(\log^3 n)$. 
 By Lemma \ref{l:sparity}, the slow parity protocol requires polynomial time. 
 By Theorem \ref{th:stable+mc+}, the expected time of the efficient stable parity protocol is $O(\log^3 n)$.
 The total number of states in the combined parity protocol is equal to the sum of their numbers of states\ $O(\log^3 n)$.
 
\end{proof}

\section{The Congruence Problem}\label{s:results}

We show that our parity protocol naturally extends to congruences modulo any integer $m$, with the efficient version following the same structure as in the parity case. We do not adapt the polynomial‑time succinct protocol of Czerner {\em et al.}~\cite{DBLP:journals/jcss/CzernerGHE24}, as its silence is unclear. Instead, %in Appendix Section~\ref{app:slowcongruence},
we present a polynomial time protocol that is both stable and silent.

\subsection{Slow, Stable and Silent Congruence Protocol}
\label{app:slowcongruence}

In this subsection, we address the {\em $m$-congruence problem}, which involves determining whether the size of a distinguished subpopulation $X$ of agents is congruent to zero modulo a fixed integer $m>1$. 
We now outline a slow but always stabilising protocol for the congruence problem modulo any integer $m\ge 2$. We use the fact that using the binary representation of $m$, one can define a sequence $m_0,m_1,\ldots,m_k$ such that $m_0=1,m_k=m$ and $k=O(\log m)$. This sequence satisfies for each $i$ one of the relations $m_{i+1}=2m_i$ or $m_{i+1}=m_i+1$ and will be the sequence of possible agent masses. At the beginning of the protocol, all agents from the set $X$ have mass $m_0=1$, and all other agents have mass $0$. 
If $m_{i+1}=2m_i$ and $i<k-1$, the interaction of two agents of mass $m_i$ leads to one of them acquiring mass $m_{i+1}$ and the other one acquiring mass $0$. When $i=k-1$, such an interaction reduces the mass of both agents to zero. If, in turn, $m_{i+1}=m_i+1$, then the interaction of agents with masses $m_i$ and $1$ leads to a result analogous to the one described before. A limitation of this approach is that, during execution, no agents with mass $1$ may be available when needed for the second interaction type. Therefore, we introduce another type of interaction, in which $m_{i+1}=m_i+1$ and an agent with mass $m_i$ interacts with an agent with mass $m_j\neq 1, j\le i$. In this case, the first agent acquires mass $m_{i+1}$, and the second mass $m_j-1$. 
If $m_j-1=m_{j-1}$, this interaction does not produce a mass outside the allowed set of agent masses. However, if $m_j=2m_{j-1}$, we must include mass $m_j-1$ in the allowed set, preserving a total of $O(\log m)$ distinct masses. To handle these extra masses, we introduce a rule: when an agent with mass $m_j - 1$ interacts with a mass $0$ agent, they transition to masses $m_{j-1}$ and $m_{j-1} - 1$, respectively. 
The protocol described above reduces the total mass of agents from $|X|$ to mass $\mu\in [0,m)$, such that $\mu\equiv_m |X|$. To inform all mass $0$ agents of the protocol’s outcome, any agent with non-zero mass interacting with a mass $0$ agent sets the latter to state $f$, indicating that the congruence does not hold. Finally, in interactions that result in a mass reduction, agents with mass $0$ receive a special state $T$. As long as they do not encounter agents with non-zero mass, they set the states of agents without mass to the state $t$, indicating that the congruence holds.

A formal description of the  slow stable population congruence protocol $\Call{SlowCongruence}{}$ follows.
This protocol uses $O(\log m)$ states and a crude analysis guarantees that it stabilises in expected time $O(n^2)$.
This protocol is used as a backup for the fast MC+ congruence protocol.

One can verify using binary encoding of $m$ that there exists a sequence $m_0,\ldots,m_k$ such that:
\begin{itemize}
    \item $m_0=1$, $m_k=m$,
    \item $k=O(\log m),$
    \item for any $i\in[0,k-1],$ either $m_{i+1}=2m_i$ or $m_{i+1}=m_i+1$.
\end{itemize}

The protocol defines the following states, for which we provide an intuitive description:
\begin{itemize}
    \item $(0,T)$: a strong opinion that the congruence relation is satisfied, 
    \item $(0,t)$: a weak opinion that the congruence relation is satisfied, 
    \item $(0,f)$: a weak opinion that the congruence relation is not satisfied,
    \item $(m_i,F)$ for each $i\in[0,k]$: the state carrying the \emph{mass} $m_i$ with a strong opinion that the congruence is not satisfied.
\end{itemize}

For each $i:0<i\le k,$ s.t., $m_i\neq m_{i-1}+1$, we add the state $(m_i-1,F)$.
In the \emph{initial} configuration:
\begin{itemize}
    \item all elements of $X$ are in  state $ (1,F)$,
    \item all other agents are in  state $(0,t)$ 
\end{itemize}

Next, we specify the transitions:
$$\begin{array}{lllll}
    (10) & (m_i,F) + (m_i,F) & \rightarrow & (m_{i+1},F) + (0, f) & \text{if } m_{i+1}=2m_i, i+1<k\\
    (20) & (m_i,F) + (1,F) & \rightarrow &(m_{i+1},F) + (0, f) & \text{if } m_{i+1}=m_i+1, i+1< k\\
    (30) & (m_{k-1},F) + (m_{k-1},F)& \rightarrow & (0,T)+(0,T) & \text{if } m_k=2m_{k-1}\\
    (40) & (m_{k-1},F)+(1,F) & \rightarrow & (0,T)+(0,T) & \text{if } m_k=m_{k-1}+1\\
    (50) & (m_i,F)+(m_j,F) & \rightarrow & (m_{i+1},F)+ (m_j-1,F) & \text{if }m_{i+1}=m_i+1, 0<j\le i\\
    (60) & (m_{i+1}-1,F) + (0, x) & \rightarrow & (m_i-1,F)+(m_i,F) & \text{if } m_{i+1}=2m_i, \\
    & & & & x\in\{T,t,f\} \text{ and }i>0\\
    (70) & (m_{k-1},F) + (m_j,F) & \rightarrow & (0,f)+(m_j-1,F) & \text{if } m_k=m_{k-1}+1, j>0\\
    (80) & (x,F) + (0,T)& \rightarrow & (x,F)+(0,f) & \text{for arbitrary } x\\
    (85) & (x,F) + (0,t)& \rightarrow & (x,F)+(0,f) & \text{for arbitrary } x\\
    
    (90) & (0,T) + (0,f) & \rightarrow & (0,T) + (0,t)
\end{array}$$

Let the \emph{mass} of an agent be defined as the value of the first coordinate of its state. A state is said to be \emph{carrying mass} if its mass is greater than zero. Furthermore, the \emph{mass} of a configuration is equal to the sum of the individual masses of all agents. The definitions of the initial configuration and the transition function lead to the following fact.

\begin{fact}\mbox{}

    \begin{enumerate}
        \item The mass of the initial configuration is equal to $|X|$.

        \item The total mass cannot increase.

        \item In one transition, the mass can either remain unchanged or decrease by $m$.

        \item The total mass cannot become negative.
    \end{enumerate}
\end{fact}

We divide the states of the form $(x,F)$ into 
\emph{main} states and \emph{auxiliary} states.
We say that a state $(m_i,F)$ for $i\in[1,k]$ is a \emph{main} state. Each state $(x,F)$ that is not main is called an \emph{auxiliary} state.
Thus auxiliary states 
have form $(m_i-1,F)$, where $m_{i}\neq m_{i-1}+1$.

Lemma~\ref{l:mass_reduction} is instrumental in establishing the time efficiency of the $\Call{SlowCongruence}{}$ protocol.

\begin{lemma}\label{l:mass_reduction}
If the total mass of the population is at least $m$, $\Call{SlowCongruence}{}$ protocol reduces this mass by $m$ in the expected time $O(mn)$.
\end{lemma}

\begin{proof}
We call a configuration \emph{$i$-configuration} when $i$ is the highest index for which there is an agent with the main state $(m_i,F)$ or the auxiliary state $(m_{i+1}-1,F)$.
We prove the following claim.

\begin{claim}\label{c:j_conf}
The expected time to reach for the first time $j$-configuration, when we start from $i$-configuration for $i< j$ is at most $1.5(m_j-1)n$.
\end{claim}

\begin{proof}[Proof of Claim]
We demonstrate our claim by induction on $j$. Consider the first appearance of a $j$-configuration with no main states $(m_j,F)$. After time of at most $n/2$, an agent in auxiliary state $(m_{j+1}-1,F)$ adopts state $(m_j,F)$ during interaction with an agent in state $(0,x)$. Two cases occur. 

In the first case, where $m_{j+1}=m_j+1$, the state transition to $(m_{j+1},F)$ occurs as soon as an agent in state $({m_j,F)}$ interacts with another agent in main state.
If there is no agent in the main state, it is created after the first interaction of an auxiliary state and state $(0,x)$, and this will take in expectation time at most $n/2$. Subsequently, there will be at least one agent in the main state other than the one in state $(m_j,F)$. In expected time at most $n/2$, an interaction between these agents will occur, which results in the adoption of state $(m_{j+1},F)$ by some agent. Thus, the total time until $(j+1)$-configuration appears is at most
\[1.5(m_j-1)n+n/2+n/2+n/2 =1.5(m_{j+1}-1)n.
\]

In the second case, where $m_{j+1}=2m_j$, for as long as there is exactly one agent in state $(m_j,F)$ in the $j$-configuration, this agent will not participate in any meaningful interactions. Therefore, we can consider a configuration of all agents, excluding a single agent in state $(m_j,F)$. By inductive assumption, this configuration becomes a $j$-configuration after an expected time of at most $1.5(m_j-1)n$.
And when this happens, it takes an expected time at most $n/2$, for a second agent to appear in state $(m_j,F)$. Interaction of two agents in this state causes one of them to adopt state $(m_{j+1},F)$ after an expected time of at most $n/2$.
Therefore, in this case, the transition to the $(j+1)$-configuration takes an expected time of at most
\[1.5(m_j-1)n+1.5(m_j-1)n+n/2+n/2+n/2=1.5(m_{j+1}-1)n.
\]
\end{proof}

The last step of the proof is estimating the time in which a mass reduction by transition (30),(40), or (70) occurs when we start from $(k-1)$-configuration.
This can be handled analogously to the two previously considered cases, distinguished by whether $m_k=m_{k-1}+1$ or $m_k=2m_{k-1}$. This concludes the proof of Lemma~\ref{l:mass_reduction}, i.e., the mass reduction is obtained in expectation after time at most $1.5(m-1)n$ starting from an arbitrary configuration.
\end{proof}

There are two possible outcomes of the congruence computation. In the first, congruence holds, and the mass of the configuration is eventually reduced to zero. In the second, the total mass is reduced to an integer value $\mu\in(0,m)$, invalidating the congruence. 
The following lemma holds.

\begin{lemma}\label{l:finalm}
 Assume that at some point the total mass $\mu\in(0,m)$. Then, all states carrying mass are distinct main states $(m_i,F)$ and, except for the state with the smallest mass, all have masses $m_i,$ where $m_{i+1}=2m_i$.
\end{lemma}

\begin{proof}
 Let $i$ be the greatest index such that $m_i\le\mu$. By Claim \ref{c:j_conf}, the state $(m_i,F)$ is formed after expected time at most $1.5(m_i-1)n+n/2$. 
 Note that if $\mu=m_i$, the agent in state $(m_i,F)$ will eventually accumulate the entire mass of $\mu$.

 If, in turn, $m_i<\mu$, then $m_{i+1}=2m_i$ must hold. After formation of the state $(m_i,F)$, the remaining mass of agents is $\mu'=\mu-m_i<m_i$.

 One can observe that the agents contributing to the mass $\mu'$ will not have any meaningful interactions with the agent in state $(m_i,F)$. Therefore, in subsequent interactions, they will form a set of agents with masses $(m_j,F)$ as in the body of the lemma. This will take the total time
 \[\sum_j (1.5(m_j-1)n+n/2)=O(mn).
 \]
\end{proof}

\pagebreak

\ \\
%\vspace{0.5cm}
\begin{algorithm}[H]
\begin{algorithmic}[1]
\Require{~$X$ a set of agents}
\Ensure{\textsc{True} if $|X|$ is divisible by $m$, otherwise \textsc{False}}
\Function{$m$-Congruence}{$X$}
\State \Call{LeaderElection}{\/}
\State \Call{$m$-WeightCreation}{X}
\State $Y \gets \emptyset$
\For{$\ell\gets LOG$ \textbf{downto} $0$} 
  \State{\Call{AddWeight}{$\ell$} \Comment{Add weight $m2^\ell$ to $Y$}} 
  \State{Compare $X$ and $Y$ using the $\Call{Majority}{}$ protocol}
  \LineIf{$|X|<|Y|$}{$\Call{RemoveWeight}{\ell}$\Comment{Remove weight $m2^\ell$ from $Y$}}
\EndFor
\LineIfElse{$|X|=|Y|$}{\Return \textsc{True}} {\Return \textsc{False}}
\EndFunction
\end{algorithmic}
\caption{$\Call{\rm\emph{m}-Congruence}{}$ Protocol}\label{alg:congruence}
\end{algorithm}
\vspace{0.5cm}
\begin{algorithm}[htb]
\begin{algorithmic}[1]
\Require{A set of agents with a single leader}
\Ensure{$m$-Weights created among non-leaders}
\Function{$m$-WeightCreation}{$X$}

\State{The leader adopts state $\lambda$ and non-leaders adopt state $\phi$.}
\For{round $r\gets 0$ \textbf{to}  $LOG+2$} 

\State{\textbf{upon} entering round $r$ interacting with an agent in state different than $m(2)$:}
\indLineIfElse{$r=0$}{\State{\hspace{\algorithmicindent}\hspace{\algorithmicindent} the leader adopts state $\lambda_0$ and non-leaders adopt state $\phi$.}}
  {\State{\hspace{\algorithmicindent}\hspace{\algorithmicindent} the leader changes state $\lambda$ to state $\lambda_r$ (with mass $m_{k+r} = m2^r$),}
  \State{\hspace{\algorithmicindent}\hspace{\algorithmicindent} a follower changes local state $w_{r-2}(0)$ to global state $w_{r-2}$ and local $w$,}
  \State{\hspace{\algorithmicindent}\hspace{\algorithmicindent} the leader changes state $\lambda_{r-1}$ to state $m(1)$,}
  \State{\hspace{\algorithmicindent}\hspace{\algorithmicindent} a follower changes state $w_{r-1}(i)$ to state $m(1)$, when $i>0$,}  
  \State{\hspace{\algorithmicindent}\hspace{\algorithmicindent} the leader changes state $m(1)$ to state $m(2)$,}
  \State{\hspace{\algorithmicindent}\hspace{\algorithmicindent} a follower in state $m(1)$ \textbf{signals}\ \textsc{Anomaly}.}}  
  
  \State{\textbf{upon} entering round $r$ interacting with an agent in state $m(2)$:}
  \State{\hspace{\algorithmicindent} a follower adopts state $m(2)$.}

  \DLineIf{interacting agents are in round $r$}{
   \hspace{\widthof{ }}The leader transfers mass $m_{k+r}$ to the first free agent it encounters:
  \State{\hspace{\algorithmicindent}\hspace{\algorithmicindent} $\lambda_r + \phi \rightarrow \lambda + w_r(k+r)$}\\
  \hspace{\algorithmicindent}\hspace{\algorithmicindent}\hspace{\algorithmicindent}\hspace{\widthof{\textbf{if}}}Distribute the mass by repeating for $i>0$: 
  \indindLineIfElse{$m_i=2m_{i-1}$}{$w_r( i) + \phi \rightarrow w_r( i-1) + w_r( i-1)$} {$w_r( i) + \phi \rightarrow w_r( i-1) + w_r(0)$}\\
  \hspace{\algorithmicindent}\hspace{\algorithmicindent}\hspace{\algorithmicindent}\hspace{\widthof{\textbf{if}}}Epidemic of state $m(1)$: $m(1) + \text{\emph{any state}} \rightarrow m(1)+m(1)$.}

\EndFor
\State{Compare $\bigcup W_k$ and $X$ using the $\Call{Majority}{}$ protocol}
\State{(elements of $\bigcup W_k$ have global states $w_k$)}
\LineIf{$|\bigcup W_k|<|X|$}{\textbf{signal}\ \textsc{Anomaly}}

\EndFunction
\end{algorithmic}
\caption{$\Call{\rm\emph{m}-WeightCreation}{}$ Protocol}\label{alg:m-weights}
\end{algorithm}
\vspace{0.5cm}

\begin{restatable}{theorem}{thslowc}\label{th:slow:alg}
The $\Call{SlowCongruence}{}$ protocol solves the $m$-congruence problem using $O(\log m)$ states and stabilises to a silent configuration in expected time $O(n^2)$.
\end{restatable}
\begin{proof}
The mass of the initial configuration is reduced to $\mu\in[0,m)$ through at most $\lfloor n/m\rfloor$ reductions, each requiring, by Lemma \ref{l:mass_reduction}, expected time $O(mn)$.
Hence, the total time of mass reduction is $O(mn\cdot \lfloor n/m \rfloor)= O(n^2)$.
Then the carrying mass agents reach their final configuration in time $O(mn)$, as established in Lemma \ref{l:finalm}.
Finally, after time $O(n\log n)$, the agents with mass 0 reach their final states, which can be analysed similarly to the Coupon Collector's Problem.
In conclusion, the total stabilisation time required to reach a silent configuration is $O(n^2)$.
The number of states is $O(\log m)$, which follows directly from the definition of the protocol.
\end{proof}

\subsection{Efficient Congruence Protocol}

An efficient MC+ congruence protocol can be constructed analogously to the efficient parity protocol. 
With this aim, we propose $m$-$\Call{Congruence}{}$ protocol (Algorithm~\ref{alg:congruence}) in which
the only remaining component to be defined is the mechanism used to generate weights with masses $m,2m,4m,8m,\ldots$, referred to as $m$-$\Call{WeightCreation}{}$ (Algorithm~\ref{alg:m-weights}).
 It differs from $\Call{WeightCreation}{}$ protocol only in a minor detail.
To enable this new subprotocol, we consider the sequence of numbers $1=m_0,m_1,m_2,\ldots,m_k=m$, which is the same as in the $\Call{SlowCongruence}{}$ protocol, and further extend this sequence with values $m_{k+i}=m2^i$, for $i>0$. And, for any integer $i>0$, we still get either $m_{i}=2m_{i-1}$ or $m_{i}=m_{i-1}+1$.

We conclude with a theorem whose proof closely parallels that of Theorem~\ref{th:parity}.

\begin{theorem}\label{th:congruence}
There exists a population protocol for the $m$-congruence problem that uses $O(\log^3 n)$ states and silently stabilises both whp and in expected $O(\log^3 n)$ time.
\end{theorem}

%\section{Relocated Pseudocodes for Fast Congruence}

\section{Conclusions}\label{s:conclusion}

We provide the first efficient stable and silent protocols for both parity and congruence predicates. They utilise $O(\log^3 n)$ states and stabilise in $O(\log^3 n)$ time, both in expectation and with high probability, advancing the research directions identified in~\cite{DBLP:journals/jcss/CzernerGHE24}. Our approach is based on a modular framework centred on a novel counting mechanism that combines {\em population weights} with {\em majority computation} within the proposed MC+ computational paradigm. While the natural lower bound on time is $\Omega(\log n)$, the only known state-space lower bound for efficient protocols is $\omega(1)$, see~\cite{BDS17}, leaving a substantial gap as an open problem.

%In this work, we introduce a modular framework for the design of efficient and stable congruence protocols. The framework is centred around a novel counting mechanism that integrates {\em population weights} with {\em majority computation} within the proposed MC+ computational paradigm. 
%Our approach builds upon prior developments in clocking, anomaly signalling, and switching mechanisms, which together ensure a reliable fallback to stable and silent protocols. Within this framework, 
%We provide the first efficient stable protocols for both parity and congruence predicates, which are moreover silent. These protocols require $O(\log^3 n)$ states and stabilise in $O(\log^3 n)$ time, both in expectation and with high probability, thereby providing significant progress towards the research directions identified in~\cite{DBLP:journals/jcss/CzernerGHE24}. 

A key feature of our construction is the expressive power of the weight system, which enables implicit conversion between unary and binary representations. This allows the parity protocol to compute and output in binary more general congruence functions, as well as exact counts of designated subpopulations or the entire population. The result can be encoded across $O(\log n)$ agents, each representing one bit, thereby resolving a postulate from~\cite{DOTY202191}.

\bibliographystyle{plainurl}
\bibliography{parity-bibliography-fixed}

\end{document}